%% file: templateArxiv.tex
\documentclass{article}

\usepackage{PRIMEarxiv}

\usepackage[utf8]{inputenc} 
\usepackage[T1]{fontenc}    
\usepackage{hyperref}       
\usepackage{url}            
\usepackage{booktabs}       
\usepackage{amsfonts}       
\usepackage{nicefrac}       
\usepackage{microtype}      
\usepackage{lipsum}
\usepackage{fancyhdr}       
\usepackage{graphicx}       
\graphicspath{{media/}}     
\usepackage{cite}
\usepackage{amsmath,amssymb,amsfonts}
\usepackage{algorithmic}
\usepackage{graphicx}
\usepackage{textcomp}
\usepackage{xcolor}
\usepackage{comment}
\usepackage{subfiles}
\usepackage{enumitem}
\usepackage{color, colortbl}
\usepackage{tikz}
\usepackage{caption}
\usepackage{subcaption}
\usepackage{multirow}
\usepackage{url}            
\usepackage{booktabs}       
\usepackage{nicefrac}       
\usepackage{amsthm}
\usepackage{makecell}

\usepackage{float}
\usepackage{algorithm}
\usepackage{algorithmic}
\usepackage{amsmath}
\usepackage{amsthm}
\newtheorem{theorem}{Theorem}
\newtheorem{lemma}{Lemma}

\title{LayerCheck: Adaptive Layer-wise Checkpointing for Large Language Model Post-training
}

\author{
Minqiu Sun \\
University of Delaware \\
Newark, United States \\
\texttt{mqsun@udel.edu}
\And
Xin Huang \\
RIKEN Center for Computational Science \\
Kobe, Japan \\
\texttt{xin.huang@riken.jp}
\And
Luanzheng Guo \\
Pacific Northwest National Laboratory \\
Richland, United States \\
\texttt{lenny.guo@pnnl.gov}
\AND
Nathan R. Tallent \\
Pacific Northwest National Laboratory \\
Richland, United States \\
\texttt{Nathan.Tallent@pnnl.gov}
\And
Kento Sato \\
RIKEN Center for Computational Science \\
Kobe, Japan \\
\texttt{kento.sato@riken.jp}
\And
Dong Dai \\
University of Delaware \\
Newark, United States \\
\texttt{dai@udel.edu}
}

\begin{document}
\maketitle

\begin{abstract}
With the rising computational and monetary costs of training large language models (LLMs), checkpointing---periodically storing model states for recovery---becomes essential for fault tolerance. Conventional checkpointing entails a severe trade-off between checkpoint frequency (I/O overhead) and computational recovery (recovery time). State-of-the-art approaches mitigate this cost through pipelining checkpoint I/Os, differential checkpointing, or in-memory persistence, yet none leverage the distinct characteristics of LLM training dynamics, where model weight updates are non-uniformly distributed across transformer layers. This observation implies that saving all weights each time might not be efficient. Inspired by this observation, we present \texttt{LayerCheck}, a layer-wise adaptive checkpointing framework that selectively persists layers whose updates exceed a threshold. This design avoids periodic I/O bursts by distributing layer-wise checkpoint writes over time, resulting in smoother and more balanced I/O profiles. Upon recovery, \texttt{LayerCheck} reconstructs a mixed-timestamp composite model state by aggregating the most recently persisted versions of each layer together with their matching optimizer states. Under a bounded per-layer staleness guard, this introduces a controlled perturbation: under standard Adam assumptions it adds a bounded staleness term, and empirically the post-restart loss deviates from the failure-free trajectory by at most \textbf{0.54\%}. Empirical results on multiple open-source LLMs with different datasets further demonstrate that recovered models preserve the original convergence behavior and accuracy while substantially reducing checkpoint overheads. Specifically, \texttt{LayerCheck} achieves up to 22.6× reduction in total checkpoint size and 1.31× reduction in end-to-end training time compared to state-of-the-art systems, significantly lowering the cost of checkpointing.
\end{abstract}


\input{./introduction}

\input{./background}

\input{./methodology}

\input{./evalv2}

\input{./related_work}

\input{./conclusion}

\section*{Acknowledgments}
We sincerely thank the reviewers for their valuable feedback. This work was supported in part by the National Science Foundation (NSF) under grants CCF-2521613 and CCF-2412345. 
This research is in part supported by the U.S.\@ Department of Energy (DOE) through the Office of Advanced Scientific Computing Research's ``Orchestration for Distributed \& Data-Intensive Scientific Exploration'' (77765) 
and Pacific Northwest National Laboratory's AT SCALE LDRD ``Decentralized data mesh for autonomous materials synthesis''.
PNNL 
is operated by Battelle for the DOE under Contract DE-AC05-76RL01830.
The Authors acknowledge the National Artificial Intelligence Research Resource (NAIRR) Pilot for contributing to this research result. This work was partially supported by also the RIKEN TRIP Initiative (AGIS / Foundational Software Development) and JSPS KAKENHI Grant Number JP24K14974.

\bibliographystyle{unsrt}  
\bibliography{sample-base}

\end{document}

%% file: introduction.tex
\section{Introduction}

Large language models (LLMs) have undergone rapid and unprecedented advancements in recent years. Training these models, including both pre-training and post-training, has become increasingly time-consuming. It's well known that pre-training, which establishes the model's foundational capabilities, can require thousands of GPUs running continuously for weeks or months~\cite{llama3}. Post-training, which is the main focus of this study, still demands hundreds of GPUs over days or weeks~\cite{instruction-finetune,powermorph}. Such prolonged runtimes render failures unavoidable. Gandhi et al.~\cite{gandhi2025moetionefficientreliablesparse} report that failures can occur as frequently as every 45 minutes in large-scale training.

Checkpointing, originally developed in HPC to protect long-running scientific applications~\cite{moody2010design, hargrove2006berkeley, fu2024autocheck, xu2024efficient, fu2024benchmarking}, is the standard fault-tolerance mechanism for large-scale LLM training~\cite{checkpointhpcapp,checkpointingforhpc}. Mainstream frameworks such as DeepSpeed~\cite{deepspeed} and FSDP~\cite{FSDP} typically use fixed-interval checkpointing, periodically saving the full training state, including model weights, optimizer states, and RNG/scheduler states, to persistent storage for failure recovery~\cite{gandhi2025moetionefficientreliablesparse}.

However, fixed-interval checkpointing introduces substantial overhead in I/O, storage, and, most importantly, end-to-end training time~\cite{guptaJustInTimeCheckpointingLow2024}. It fundamentally trades checkpoint overhead for recovery cost: infrequent checkpoints reduce runtime overhead but increase recomputation after failures, whereas frequent checkpoints shorten rollback distance but incur higher I/O and synchronization cost. This trade-off worsens at scale as job-level MTBF decreases, forcing shorter checkpoint intervals. Prior work reports that checkpointing can consume up to 12\% of total training time, and as much as 43\% in extreme cases~\cite{maengCPRUnderstandingImproving}.

Recent work on LLM checkpointing overhead falls into three groups.
(1) \textbf{Pipeline optimization.} GEMINI~\cite{wangGEMINIFastFailure2023} uses in-memory checkpointing to overlap GPU–storage transfers with training, enabling high-frequency snapshots. DataState-LLM~\cite{Datastate-LLM} further decouples checkpointing from computation via lazy synchronization. These approaches, however, struggle as model sizes grow and the GPU–storage gap widens.
(2) \textbf{Offline compression.} Delta-DNN~\cite{DeltaDNN} and ExCP~\cite{liexcp} apply lossy differential compression across consecutive checkpoints to cut storage cost. However, full checkpoints must still be materialized before compression, leaving I/O volume unchanged.
(3) \textbf{Selective persistence.} Online differential checkpointing persists only changes: Check-N-Run~\cite{check-n-run} targets recommendation models, while LowDiff~\cite{lowdiff} extends the idea to LLMs by saving only per-iteration gradients, at the cost of recomputation and lineage tracking. Amber~\cite{Amber} further explores selective incremental checkpointing for LLM training by using a bitmap to identify changed parameters and persist only those parameters together with their optimizer state. While this can reduce written bytes, it also requires fine-grained per-parameter tracking and bitmap maintenance, whose metadata and runtime overhead grows with model size. All of these still stall training and inflate recovery cost.

Existing selective/differential persistence schemes typically decide what to write based on a global rule (e.g., fixed frequency, fixed partitions, or explicit deltas). In contrast, we exploit an empirical property that is repeatedly observed in LLM post-training: \emph{different transformer layers evolve at different paces across steps}~\cite{temprature}. This layer-wise update asynchrony suggests that writing the \emph{entire} model at every checkpoint is often redundant: at many steps, only a subset of layers experiences meaningful drift.

Motivated by this observation, we propose \texttt{LayerCheck}, a \emph{layer-wise adaptive checkpointing} mechanism. \texttt{LayerCheck} tracks the update significance of each layer during training and persists only those layers whose accumulated drift is sufficiently large. By writing partial model states in most checkpoints, \texttt{LayerCheck} substantially reduces checkpoint I/O volume while keeping recovery practical: reconstruction simply merges the most recently persisted state of each layer into a composite checkpoint.

Selective persistence introduces an inherent trade-off at recovery time: layers not written in the most recent checkpoint are restored from earlier iterations, yielding a \emph{mixed-timestamp} model state. Following the perturbation-based view of partial recovery~\cite{partialrecover}, we interpret this mismatch as a bounded perturbation from the ideal failure-free trajectory. Let $t_i$ denote the last iteration at which layer $i$ is persisted, and let $\tau$ be the failure iteration. We define the \emph{layer age} (staleness) of layer $i$ at recovery as $\tau - t_i$. To prevent any layer from becoming arbitrarily stale, \texttt{LayerCheck} enforces a staleness guard with a user-specified budget $S_{\max}$: if $\tau - t_i \ge S_{\max}$, the system forces persistence of layer $i$ even when its drift is below the update threshold. Under this bounded-staleness invariant, we show that recovery does not change the asymptotic convergence rate (\S\ref{sec:Method}), and we further validate empirically that post-recovery training closely matches full-checkpoint baselines (\S\ref{sec:evaluation}).

\texttt{LayerCheck} is guided by three practical design principles. \emph{First,} it makes persistence decisions based on \emph{accumulated} layer drift over multiple iterations, so that gradual but meaningful changes are not missed. \emph{Second,} instead of tracking per-parameter updates, it maintains a lightweight \emph{per-layer} update statistic computed from already-available training signals, reducing the tracking cost to a small vector of layer-level summaries. \emph{Third,} it uses a threshold $K$ tied to the magnitude of training updates, optionally in a learning-rate-aware manner, while the staleness guard $S_{\max}$ provides an explicit safeguard independent of threshold tuning.

We evaluate \texttt{LayerCheck} on multiple open-source LLMs and measure checkpoint I/O volume, storage footprint, end-to-end training time, and recovery time under failure injection. Compared to state-of-the-art checkpointing systems, \texttt{LayerCheck} reduces the \emph{total} checkpoint size by up to \textbf{22.6×} and improves total training time by up to \textbf{1.31×} without measurable degradation in model quality in our evaluated settings. Recovery is also faster: by reconstructing the composite state from the most recently persisted layers, \texttt{LayerCheck} reduces recovery cost by up to \textbf{3.4×}. In addition, we further study projected behavior under more failure-prone scale-out regimes using a simulation parameterized by measured system components. These results indicate that layer-wise selective persistence can deliver substantial system-level savings without sacrificing training outcomes. Our contributions are threefold:
\begin{itemize}[leftmargin=1.5em,nosep]
  \item \textbf{Layer-wise selectivity from training dynamics.}
  We identify \emph{layer-wise update asynchrony} as an exploitable property of LLM post-training and propose \texttt{LayerCheck}, a layer-wise selective checkpointing mechanism that persists full layer states only when their accumulated updates are significant, with negligible overhead.

  \item \textbf{Recovery fidelity under mixed-timestamp checkpoints.}
  We provide theoretical and empirical evidence that LayerCheck's mixed-timestamp checkpoints preserve training fidelity: post-restart loss deviates from the failure-free trajectory by at most \textbf{0.54\%}, and downstream benchmark scores match failure-free fine-tuning.

  \item \textbf{Fast, practical recovery without replay.}
  We design a per-layer checkpointing and reconstruction workflow that restores the composite, full checkpoint effectively, avoiding the retraining and replay overheads common in fixed-frequency checkpointing.

\end{itemize}

The rest of the paper is organized as follows. \S\ref{sec:background} introduces the necessary background. In \S\ref{sec:Method} we describe the design and implementation of \texttt{LayerCheck}. We present the extensive experimental results in \S\ref{sec:evaluation}. \S\ref{sec:related} discusses the closely related work. \S\ref{sec:conclude} concludes this study and discusses future work.

The source code of \texttt{LayerCheck} is publicly available at \url{https://github.com/DIR-LAB/LayerCheck}.

%% file: background.tex
\section{Background}
\label{sec:background}
Large language models (LLMs) have demonstrated remarkable power in various domains~\cite{naveed2025comprehensive,egersdoerfer2025stellar,egersdoerfer2025ioagent,egersdoerfer2024ion,pfagenttractableselfevolvingpowerflow}. 
LLM training, like other neural networks, involves a forward pass, loss computation, a backward pass, and parameter update. Here, we review the key fundamentals that motivate our layer-wise checkpointing: (i) the structural components of LLMs, (ii) the unbalanced updating characteristics intrinsic to LLM training.

\subsection{Problem Setup and Goals}
\label{sec:problem_setup}
We consider large language model (LLM) post-training under a fail-stop failure model, where the job may terminate unexpectedly and must be restored from persisted checkpoints.
A checkpoint includes the model weights, any persistent optimizer state associated with them, when present (e.g., Adam/AdamW momentum), and some other configuration files of the model's architecture, training arguments, and random number generator, which are necessary to resume training.
Unlike traditional checkpointing that snapshots a globally consistent state at a single iteration, \texttt{LayerCheck} allows \emph{mixed-timestamp} restoration, where each layer may be restored from its most recent available checkpoint.
Our goal is to reduce end-to-end training time under failures by reducing checkpoint I/O overhead while preserving post-recovery training behavior and final model quality.

\subsection{Layer Structure of LLMs}
In Figure~\ref{fig:llama_workflow}, we show the layer structure of Qwen2.5-7B model~\cite{qwen2025qwen25technicalreport}. Here, we partition trainable parameters into N layer groups at the granularity of modules: \texttt{embed\_tokens}, each $\texttt{transformer block}_i$ ($i=1,\ldots,B$), \texttt{normalization}, and \texttt{lm\_head}, yielding $N=B+3$. Tokens are first embedded to a vector, then propagated through 28 consecutive transformer layers. Then, after layer normalization to eliminate internal covariate shift, the hidden representations are projected back to the vocabulary space to produce logits in the \texttt{lm\_head} layer; finally, a softmax yields output probabilities. Note that in smaller models, the \texttt{lm\_head} layer may be weight-tied to \texttt{embed\_tokens} to reduce parameter count \cite{weight_tying}.

\begin{figure}[ht!]
    \centering
\includegraphics[width=0.8\linewidth]{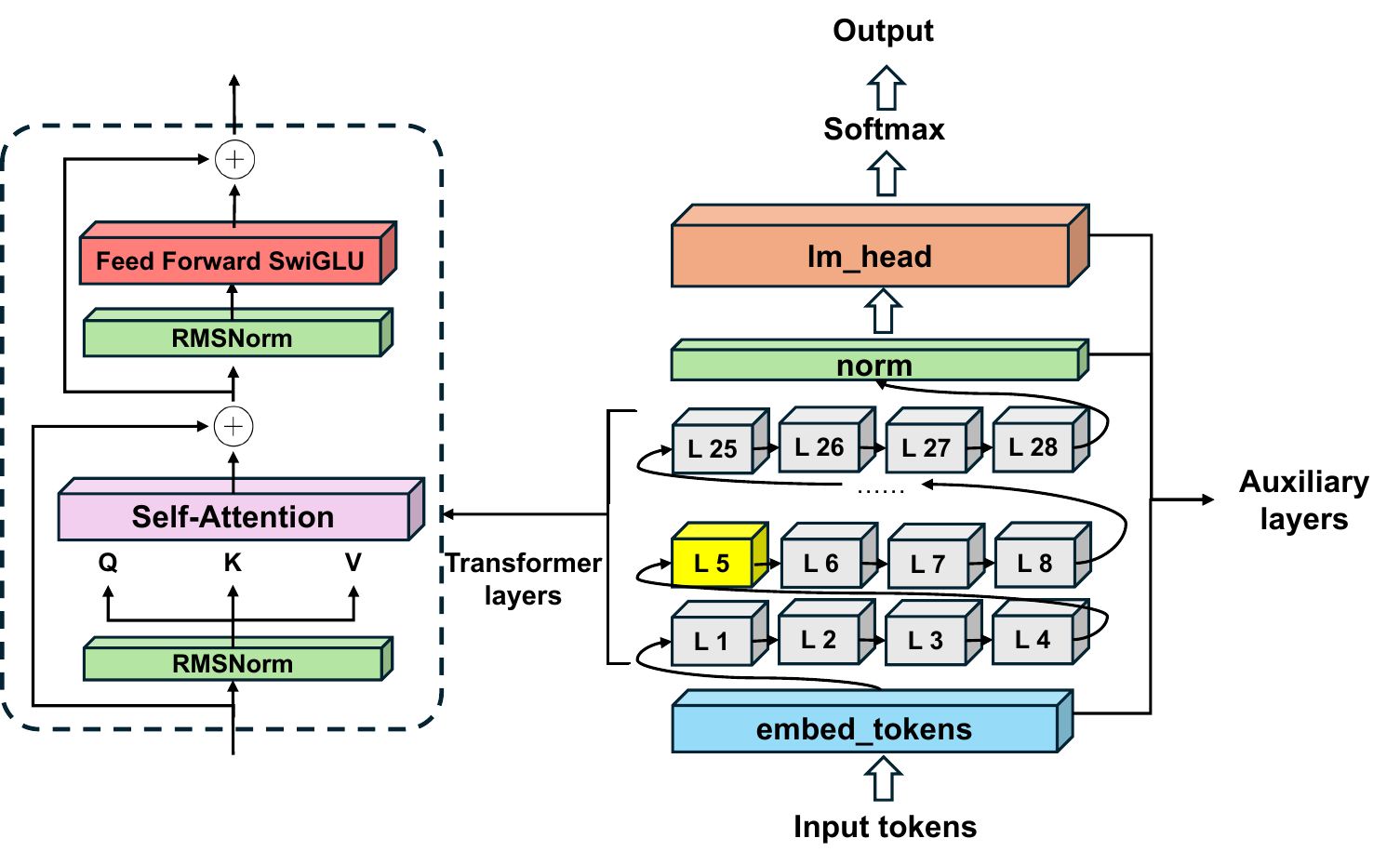}
    \caption{The layer-wise structure in a Qwen2.5-7B model.}
    \label{fig:llama_workflow}
\end{figure}


\subsection{Layer-wise Unbalanced Updating}

\begin{figure}[ht!]
    \centering
    \includegraphics[width=0.8\linewidth]{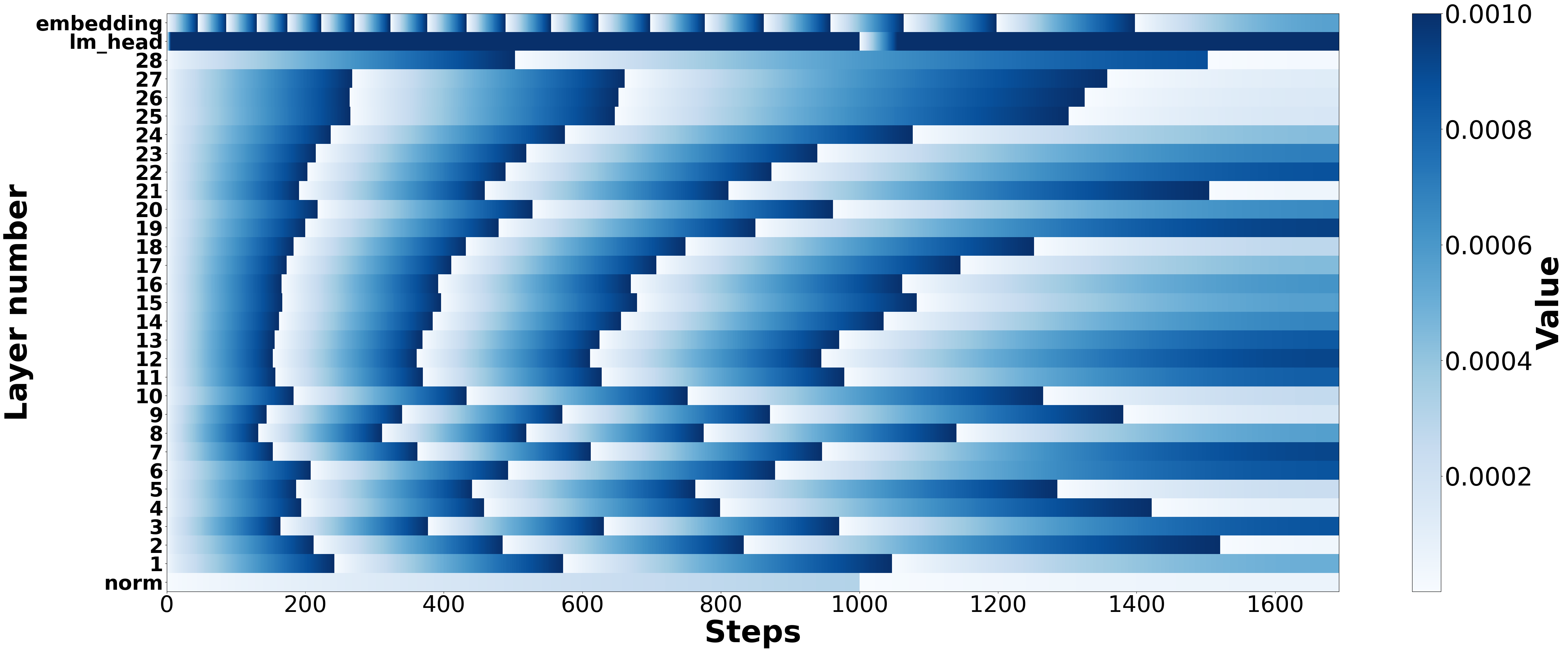}
    \caption{Temporal and layer-wise heterogeneity of model updates.}
    \label{fig:visual}
\end{figure}

Many existing studies have observed the unbalanced updates to different layers of the LLM models at different training phases~\cite{jawahar-etal-2019-bert,phang-etal-2021-fine,temprature}.
To better showcase that, we visualize the per-layer weight updates for an open-source Qwen2.5-7B model in Figure~\ref{fig:visual}. Here, the $x$-axis represents training steps, while the $y$-axis enumerates the model’s all 31 layers: 28 transformer layers and 3 auxiliary layers. The heatmap color intensity corresponds to the weight updates: darker regions indicate larger accumulated weight updates. 

Here, each layer’s accumulated weight update is reset to 0 once it exceeds the threshold, as if the layer had been checkpointed. The wave-like visualization clearly highlights the variation in update magnitudes across layers and steps, underscoring the non-uniformity of layer-wise learning dynamics during training.

Despite the substantial variability in actual weight updates across layers and steps, current checkpointing strategies uniformly save all layers. This uniform approach could be inefficient, as the most significant model updates may happen asynchronously across layers. Such unbalanced updates motivate our layer-wise checkpointing strategy.

%% file: methodology.tex
\section{Design and Implementation}\label{sec:Method}

\begin{figure*}[ht!]
    \centering
    \includegraphics[width=0.78\linewidth]{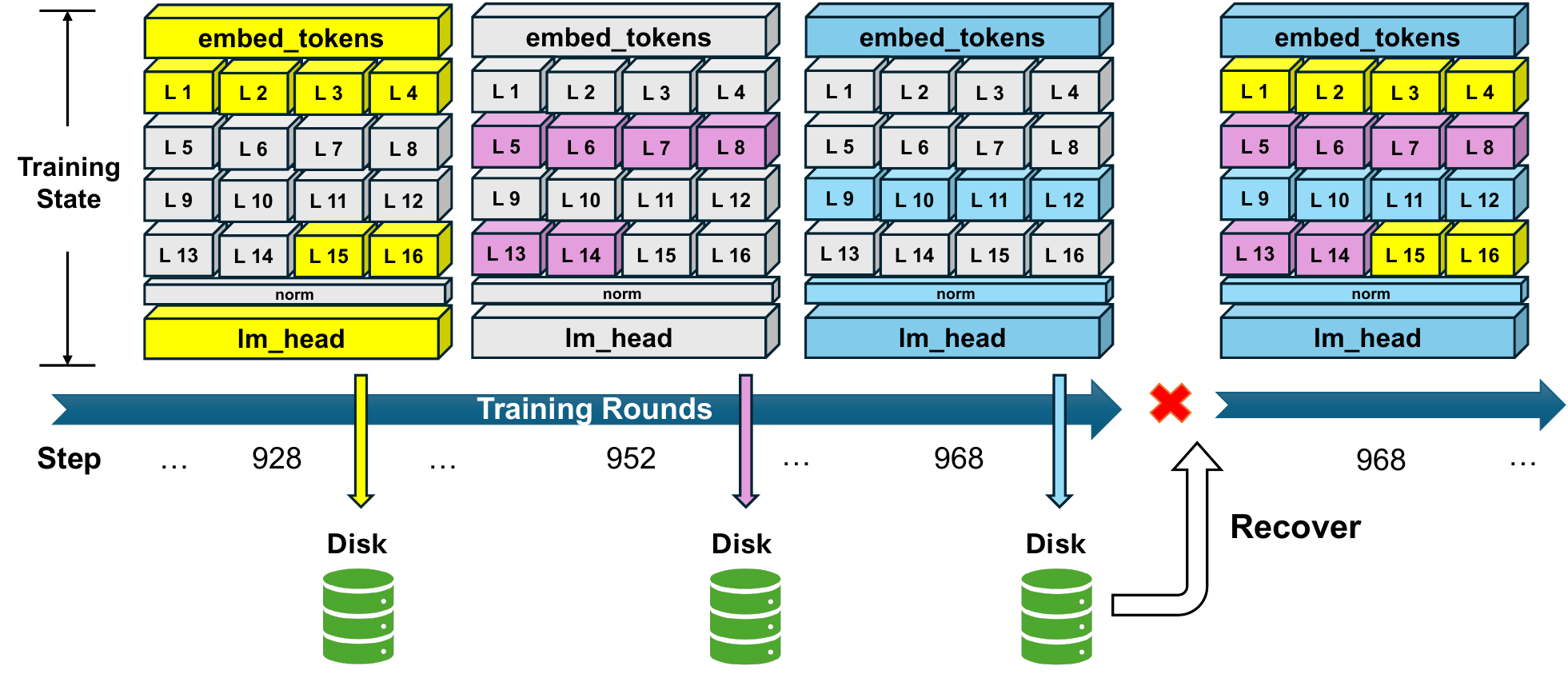}
    \caption{\texttt{LayerCheck} workflow. Layers are persisted at different iterations based on accumulated drift; recovery assembles the most recent persisted state of each layer into a composite checkpoint.}
    \label{fig:workflow}
\end{figure*}

This section presents the design and implementation of \texttt{LayerCheck}, a layer-wise selective checkpointing mechanism for LLM post-training.
\texttt{LayerCheck} is motivated by a repeatedly observed training-system property: different Transformer layers drift at different paces, making full-model snapshots frequently redundant.

At a high level, \texttt{LayerCheck} operates in three stages.
(1) \emph{Online monitoring:} during training, it maintains a lightweight per-layer drift statistic and decides which layers to persist at each checkpoint opportunity.
(2) \emph{Selective persistence:} when a layer is selected, \texttt{LayerCheck} writes the \emph{full layer state}, including both weights and optimizer tensors, so that recovery does not require delta replay.
(3) \emph{Composite recovery:} after a failure, it reconstructs a complete training state by merging the most recently persisted version of each layer, which may originate from different iterations.

The central fidelity concern is \emph{mixed-timestamp recovery}: selective persistence yields a checkpoint where different layers (and their Adam moments) can be restored from different steps.
To control this effect, we enforce a simple bounded-staleness invariant ($S_{\max}$), which also underpins the convergence argument in \S\ref{sec:convergence_layercheck}.

\subsection{Overall Workflow}
\label{subsec:overall-workflow}

Figure~\ref{fig:workflow} illustrates the workflow along the training timeline.
The $x$-axis is the iteration index.
Each block corresponds to the full training state at that iteration, while each box within the block represents a layer state consisting of (i) the layer parameters and (ii) the optimizer tensors associated with those parameters (e.g., Adam first/second moments under AdamW).

\textbf{From full snapshots to per-layer persistence.}
Conventional checkpointing persists the entire model and optimizer state at each checkpoint.
In contrast, \texttt{LayerCheck} evaluates every layer at runtime and persists only a subset of layers at most checkpoints.
A layer is selected when its \emph{accumulated normalized drift} since its last persistence exceeds a threshold.
Intuitively, fast-changing layers are refreshed more frequently, while slow-changing layers are skipped, reducing checkpoint I/O and storage.

\textbf{Composite checkpoint at recovery (no rollback).}
When a failure occurs after iteration $t_{\mathrm{fail}}$ completes, \texttt{LayerCheck} recovers by selecting, for each layer $\ell$, the most recent successfully persisted layer checkpoint at time $t_\ell \le t_{\mathrm{fail}}$ and assembling these layer states into a full training checkpoint.
This composite checkpoint can contain mixed timestamps across layers; however, because each layer is restored together with its corresponding optimizer tensors, the recovered optimizer state remains \emph{internally consistent within each layer}.
Training then resumes directly at iteration $t_{\mathrm{fail}}{+}1$ without rolling back to an older full-model snapshot.

\subsection{Layer-wise Monitoring and Checkpointing}
\label{subsec:monitoring}

This subsection details how \texttt{LayerCheck} quantifies per-layer update significance and triggers layer persistence online.

\textbf{Accumulated Per-layer Updates.} 
Checkpointing serves fault tolerance, so our goal is to ensure that resuming from a recovered composite checkpoint yields statistically indistinguishable model quality compared to a failure-free run.
Identifying, at checkpoint time, which individual parameter updates matter most is generally infeasible.
Moreover, tracking per-parameter changes would require auxiliary state comparable to the model size.
\texttt{LayerCheck} therefore tracks drift at the granularity of layers.

For each layer $\ell$, we maintain an accumulator $\Delta_\ell$ initialized to $0$.
At iteration $t$, after the optimizer computes the parameter update, we estimate a normalized update magnitude:
\begin{equation}
    u_\ell(t)
    =
    \frac{\mathrm{mean}\bigl(\lvert W_\ell(t) - W_\ell(t-1) \rvert\bigr)}
         {\mathrm{mean}\bigl(\lvert W_\ell(t-1) \rvert\bigr) + \varepsilon_w},
\end{equation}
where $W_\ell(t)$ denotes the parameter tensor(s) of layer $\ell$ at iteration $t$, $\mathrm{mean}(\cdot)$ is taken element-wise across that layer's parameters, and $\varepsilon_w$ is a small constant for numerical stability.
We accumulate drift over time:
\begin{equation}
    \Delta_\ell \leftarrow \Delta_\ell + u_\ell(t).
\end{equation}
We use the absolute value to obtain an upper-bound style estimate (conservative counting): it avoids cancellation across coordinates and biases toward persisting slightly more layers rather than missing meaningful drift.
Using the mean reduces sensitivity to outliers and yields a stable layer-level signal.

\textbf{Thresholding with learning-rate adaptation.}
A layer is persisted when its accumulated drift exceeds a threshold:
\begin{equation}
    \Delta_\ell \ge K(t).
\end{equation}
Instead of a single fixed threshold, \texttt{LayerCheck} uses a base threshold $K_0$ and scales it with the current learning rate to account for common LR schedules (warmup/decay) under AdamW.
Concretely, we set $K(t) = K_0 \cdot s(\mathrm{lr}(t))$, where $s(\cdot)$ is a lightweight scaling function derived from the LR schedule (e.g., proportional to the current LR relative to the initial LR).
When a layer is persisted, \texttt{LayerCheck} writes the full layer state and resets the accumulator:
\begin{equation}
    \Delta_\ell \leftarrow 0.
\end{equation}

\subsection{Bounding Mixed-Timestamp Staleness}

\label{subsec:staleness}

Selective persistence implies that at recovery time, different layers may be restored from different iterations.
We quantify this effect via \emph{staleness}.
If a failure occurs at iteration $\tau$, and layer $\ell$ was last persisted at iteration $t_\ell \le \tau$, then the layer staleness is $\tau - t_\ell$.
We additionally define an overall (size-weighted) checkpoint staleness metric for reporting, where larger layers contribute proportionally more.

While threshold-based persistence typically keeps most layers fresh, some layers (e.g., normalization-related tensors) can drift very slowly and may otherwise remain unpersisted for long periods.
To prevent pathological cases and to make recovery robust, \texttt{LayerCheck} enforces a staleness guard:
\begin{equation}
\label{eq:staleness_invariant}
\max_{\ell}(\tau - t_\ell) \le S_{\max}.
\end{equation}
If a layer reaches staleness $S_{\max}$, \texttt{LayerCheck} forces persistence of that layer at the next checkpoint opportunity even if its drift accumulator has not crossed the threshold, and then resets its accumulator $\Delta_\ell$ to 0. This guard applies to the \emph{maximum per-layer staleness} and serves as a worst-case safety bound. In practice, however, the checkpoint-wide mismatch is better reflected by a parameter-size-weighted average staleness, since a large maximum can be caused by a very small tensor (e.g., normalization layer) whose contribution to the full checkpoint is negligible. In our experiments, with $K_0=10^{-3}$, most layers are naturally refreshed roughly every $\sim 200$ iterations, the parameter-weighted average staleness remains low, and the bound $S_{\max}=1000$ acts as a rarely triggered safety backstop. It is shown in Figure~\ref{fig:staleness} that the maximum (blue line) is shaped by the normalization layer mainly from step 200 to step 1000, but the parameter-size-weighted average staleness (red line) reflects the mean staleness of all the layers better.
This invariant is also the key system property used in the convergence analysis (\S\ref{sec:convergence_layercheck}).

\begin{figure}[ht!]
    \centering
    \includegraphics[width=0.8\linewidth]{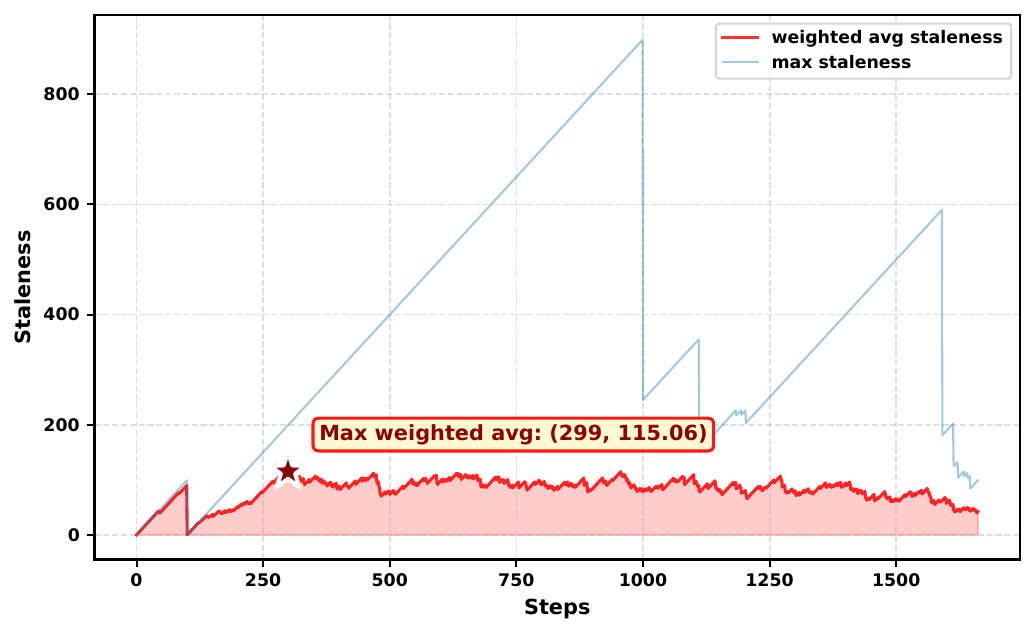}
    \caption{Aggregate staleness during Qwen2.5-7B fine-tuning with $K_0=10^{-3}$. The annotation marks the maximum parameter-weighted average staleness.}
    \label{fig:staleness}
\end{figure}

\subsection{Implementation on PyTorch/DeepSpeed}
\label{subsec:implementation}

We implement \texttt{LayerCheck} on top of PyTorch and DeepSpeed ZeRO-3~\cite{deepspeed}.
The design goal is \emph{non-intrusive integration}: between failures, training executes the standard AdamW update; \texttt{LayerCheck} only augments the optimizer with lightweight drift accounting and selectively filters what is serialized during checkpointing.

\subsubsection{Layer-aligned optimizer state layout.}
A practical challenge is that optimizer states are stored as flattened tensors and parameter groups, which do not directly expose a layer boundary.
To enable per-layer persistence and recovery, we reconstruct the optimizer parameter groups so that they mirror the model's layer-wise hierarchy while preserving the original weight-decay policy.
In our implementation, each Transformer layer is split into two parameter groups (with-weight-decay and without-weight-decay), while auxiliary layers that contain only one type are mapped to a single group.
Because this grouping order is consistent given the model configuration (number of Transformer blocks and weight-decay settings), we can locate the optimizer shard corresponding to a layer during serialization and reconstruction.

\subsubsection{Reusing the optimizer update to compute drift.}
Computing $W(t)-W(t-1)$ naively would introduce extra tensor copies. However, we piggyback drift computation on the optimizer update path to avoid extra tensor copies and keep monitoring overhead low.

\subsubsection{Threshold defaults and special-casing embeddings.}
With LR-adaptive thresholding, we set the base threshold to $K_0=10^{-3}$ for most layers.
This scale is consistent with prior observations that $10^{-3}$-level relative parameter drift is a meaningful unit during fine-tuning~\cite{Amber,taskvectorquantization}.
Embedding layers often exhibit smaller updates; to avoid forcing frequent persistence of embeddings, we use a larger base threshold for embeddings (e.g., $10^{-2}$).

We also disable selective persistence during an initial warm-up window (first 200 steps) and take a full base checkpoint that serves as a fallback for rarely-updated layers.

\subsection{Recovery from \texttt{LayerCheck} Checkpoints}

When a failure occurs at step $t_{\mathrm{fail}}$, \texttt{LayerCheck} reconstructs a full checkpoint in three steps.
\begin{enumerate}[leftmargin=1.5em,nosep]
    \item For each layer $\ell$, find the most recent persisted layer checkpoint timestamp $t_\ell \le t_{\mathrm{fail}}$.
    \item Load the corresponding layer state \mbox{$(W_\ell(t_\ell),\, O_\ell(t_\ell))$}, where $O_\ell$ denotes the optimizer tensors associated with $W_\ell$ (e.g., Adam moments). If a layer has no prior layer checkpoint, fall back to the base full checkpoint at $t_{\mathrm{base}}$.
    \item Assemble a complete training state by placing each loaded layer state back into its original layout.
\end{enumerate}
Training then resumes from the reconstructed composite checkpoint.

\textit{Metadata-driven reconstruction.}
During training, \texttt{LayerCheck} records the persistence history of each layer in each step. Upon failure, recovery derives a per-layer reconstruction plan from this metadata and assembles the most recent valid layer states into a DeepSpeed-compatible full checkpoint.

\subsection{Convergence Analysis of \texttt{LayerCheck}}
\label{sec:convergence_layercheck}

Recovery replaces $(\theta_\tau,m_\tau,v_\tau)$ with a mixed-timestamp composite state. Following ExCP~\cite{liexcp}, we provide a proof-of-concept analysis showing that bounded staleness adds a bounded term to Adam's standard regret bound without changing its asymptotic average-regret rate.

\subsubsection{Setup and assumptions} We follow the standard Adam regret analysis of Kingma \& Ba~\cite{adam}, which assumes a convex per-iteration loss. Let $\theta_t\in\mathbb{R}^d$ be the parameter vector, $f_t(\cdot)$ be the per-iteration loss, and $g_t=\nabla f_t(\theta_t)$. Adam maintains $m_t$ and $v_t$ using $\beta_1,\beta_2\in(0,1)$.
We also adopt standard boundedness assumptions used in Adam analyses~\cite{adam}:
\begin{equation}
\label{eq:bounded_grad}
\|g_t\|_2\le G,\qquad \|g_t\|_\infty\le G_\infty,\quad \forall t,
\end{equation}
\begin{equation}
\label{eq:bounded_iterates}
\|\theta_n-\theta_m\|_2 \le D,\quad \|\theta_n-\theta_m\|_\infty \le D_\infty,\quad \forall m,n\in\{1,\dots,T\}.
\end{equation}
We assume that the recovered composite state $\tilde{\theta}_\tau$ lies
in the same bounded domain and use the standard Adam initialization
$m_0=v_0=0$.

Following ExCP~\cite{liexcp}, we use $\alpha_t=\alpha/\sqrt{t}$, $\beta_{1,t}=\beta_1\lambda^{t-1}$ for $\lambda\in(0,1)$, and assume $\beta_1^2/\sqrt{\beta_2}<1$. For simplicity, we omit hats and absorb Adam's bounded bias-correction factors into the constants.

\subsubsection{Composite recovery model}
Suppose a failure occurs at iteration $\tau$.
For each layer $\ell\in\{1,\dots,L\}$, let $t_\ell\le\tau$ be the iteration of its most recent persisted layer checkpoint.
Because \texttt{LayerCheck} stores weights together with their layer-local optimizer tensors, recovery forms
\begin{equation}
\label{eq:composite_state}
\tilde{\theta}_\tau \triangleq \bigoplus_{\ell=1}^{L} \theta^{(\ell)}_{t_\ell},\quad
\tilde{m}_\tau \triangleq \bigoplus_{\ell=1}^{L} m^{(\ell)}_{t_\ell},\quad
\tilde{v}_\tau \triangleq \bigoplus_{\ell=1}^{L} v^{(\ell)}_{t_\ell}.
\end{equation}
Let $t(i)$ denote the timestamp of the layer containing coordinate $i$.
Under the staleness guard, $\tau-t(i)\le S_{\max}$ for all $i$.

\begin{lemma}[Second-moment deviation over $S_{\max}$ steps]
\label{lem:v_deviation}
Assume \eqref{eq:bounded_grad} and $v_0=0$. For any
$s\in[\tau-S_{\max},\tau]$ and every coordinate $i$,
\begin{equation}
\label{eq:v_deviation}
|v_{\tau,i}-v_{s,i}|
\le
2(1-\beta_2^{S_{\max}})G_\infty^2.
\end{equation}
\end{lemma}

\begin{proof}
Let $k=\tau-s\le S_{\max}$. Unrolling
$v_{t,i}=\beta_2v_{t-1,i}+(1-\beta_2)g_{t,i}^2$ and using
$0\le v_{s,i}\le G_\infty^2$ gives
\[
|v_{\tau,i}-v_{s,i}|
\le 2(1-\beta_2^k)G_\infty^2
\le 2(1-\beta_2^{S_{\max}})G_\infty^2.
\]
\end{proof}
The remaining recovered states are also bounded:
\begin{equation}
\label{eq:theta_m_deviation}
\|\tilde{\theta}_\tau-\theta_\tau\|_\infty\le D_\infty,\qquad
\|\tilde{m}_\tau-m_\tau\|_\infty\le2G_\infty.
\end{equation}

\begin{theorem}[Proof-of-concept regret bound under bounded staleness]
\label{thm:layercheck_rate}
Assume \eqref{eq:bounded_grad}--\eqref{eq:bounded_iterates} and the staleness guard \eqref{eq:staleness_invariant}.
Consider one recovery at iteration $\tau$ and denote the regret of the recovered run by $\tilde{R}(T)$.
Following the ExCP-style proof-of-concept decomposition, we retain the standard bounded-domain and moment terms for $\theta$ and $m$ and isolate the additional second-moment term:
\begin{equation}
\label{eq:regret_layercheck}
\tilde{R}(T) \le R_{\texttt{Adam}}(T) + \Delta_{\mathrm{stale}}(T),
\end{equation}
where $R_{\texttt{Adam}}(T)=O(\sqrt{T})$. Applying
Lemma~\ref{lem:v_deviation} and
$|\sqrt{a}-\sqrt{b}|\le\sqrt{|a-b|}$ gives
\begin{equation}
\label{eq:delta_stale}
\begin{aligned}
\Delta_{\mathrm{stale}}(T)
&\triangleq
\frac{D^2\sqrt{T}}{2\alpha(1-\beta_1)}
\sum_{i=1}^{d}
\left|\sqrt{v_{\tau,i}}-\sqrt{v_{t(i),i}}\right| \\
&\le
\frac{D^2dG_\infty}{2\alpha(1-\beta_1)}
\sqrt{2T(1-\beta_2^{S_{\max}})}
=O(\sqrt{T}).
\end{aligned}
\end{equation}
Consequently, the average regret satisfies $\tilde{R}(T)/T=O(1/\sqrt{T})$.
\end{theorem}

\subsubsection{Multiple failures.}
If recoveries occur at iterations $\tau_1<\cdots<\tau_K\le T$, the corresponding staleness terms add up. When $K$ is fixed independently of $T$, their sum remains $O(\sqrt{T})$, so the asymptotic average-regret rate is unchanged.

%% file: evalv2.tex
\section{Evaluation}
\label{sec:evaluation}
This section evaluates \texttt{LayerCheck} against the paper's central claim: \emph{layer-wise selective persistence can reduce checkpoint cost while maintaining safe and effective recovery}, even though the recovered training state is assembled from layers persisted at different timestamps. We therefore first examine whether \texttt{LayerCheck}'s mixed-timestamp recovery preserves stable post-failure training behavior.
At system-level, checkpoint-based fault tolerance incurs two kinds of overhead: checkpoint overhead during training and recovery overhead after a failure.
The former affects steady-state training efficiency, while the latter determines how much time is lost before useful training resumes.
Accordingly, we evaluate \texttt{LayerCheck} from two perspectives: the checkpoint cost it introduces during training and the recovery cost it incurs after failures. We focus on end-to-end training impact under I/O pressure and on recovery behavior under comparable recovery quality.
We organize experiments around three questions.
\begin{itemize}[leftmargin=1.5em]
    \item First, when a failure occurs and we restart from a composite (mixed-timestamp) checkpoint, does training remain stable and converge to the same quality as failure-free training?
    \item Second, in steady state, how much wall-clock time and storage can we save by avoiding unnecessary writes of slow-changing layers?
    \item Third, after a failure, how much time do we save in recovery once we control for the amount of rollback/staleness each method permits?
\end{itemize}

\subsection{Experimental Setup}
\label{subsec:eval-setup}
\subsubsection{Testbed}
All experiments run on a single 8-GPU server with NVIDIA A100 40GB GPUs and two AMD EPYC 7713 CPUs (2{,}048~GB DRAM).
The storage backend is a 100~TB Lustre parallel file system. And our evaluated node uses one Mellanox ConnectX-6 InfiniBand HCA, connected at HDR 200 Gbps.

We enable DeepSpeed ZeRO Stage-3 so that model states and optimizer tensors are sharded across GPUs.
All methods use AdamW~\cite{adamw}.
For uniform post-restart evaluation, each checkpoint additionally stores a BF16 copy of the model weights. This copy does not change \texttt{LayerCheck}'s selective persistence decisions; it simply ensures that evaluation reads the same representation across methods.

We evaluate three open-source LLMs spanning small to medium scale: Llama-3.2-1B~\cite{llama3}, Qwen-2.5-3B, and Qwen-2.5-7B~\cite{qwen2025qwen25technicalreport}.
We consider two supervised fine-tuning (SFT) workloads: MedQA~\cite{medqa} (medical domain) and OpenThoughts~\cite{guha2025openthoughtsdatarecipesreasoning} (reasoning).
Unless otherwise stated, each run trains for 1 epoch with maximum gradient norm 1 and initial learning rate $10^{-5}$.
Due to GPU memory constraints, we set sequence length to 2048, micro-batch size to 1, and gradient accumulation steps to 2 across all methods.

We compare \texttt{LayerCheck} against CheckFreq~\cite{checkfreq}, GEMINI~\cite{wangGEMINIFastFailure2023}, and LowDiff~\cite{lowdiff}.
These baselines reflect three distinct approaches: overlap checkpoint process with model training(CheckFreq), overlapped checkpoint and in-memory checkpoint (GEMINI), and reducing bytes written via differential information (LowDiff). Together with \texttt{LayerCheck}, they cover different points in the checkpointing design space and are therefore compared as representative standalone systems for frequent checkpointing/recovery settings. Unless otherwise stated, we follow default configurations in the original papers.

In all experiments, ``checkpoint time'' refers to wall-clock time spent on the checkpointing path and cannot be overlapped with other processes, including serialization, storage I/O, and any synchronization required by the framework to produce a consistent checkpoint.
This quantity directly affects training throughput, because time spent in checkpointing cannot be used for model training.

\subsubsection{Evaluation Methodology}
Under DeepSpeed ZeRO-3, model weights and optimizer states are sharded across ranks, and each rank performs checkpointing independently. The core behavior of LayerCheck — per-layer drift tracking, selective persistence, and composite recovery — is therefore determined by per-rank activity, enabling it to be faithfully evaluated in a single-node sharded setting where each rank's state mirrors the per-rank size in larger-scale parallel training. We adopt this methodology following recent LLM checkpointing systems~\cite{Amber} and extend to scale-out regimes via a data-driven simulation in~\ref{subsec:eval-scalability}, consistent with GEMINI~\cite{wangGEMINIFastFailure2023}.

We evaluate recovery fidelity using both post-restart loss trajectories and final downstream benchmark scores.

\subsection{Recovery Fidelity: Training After Mixed-Timestamp Recovery}
\label{subsec:eval-correctness-loss}

\begin{figure*}[t]
    \centering
    \begin{subfigure}[b]{0.45\textwidth}
        \centering
        \includegraphics[width=\linewidth]{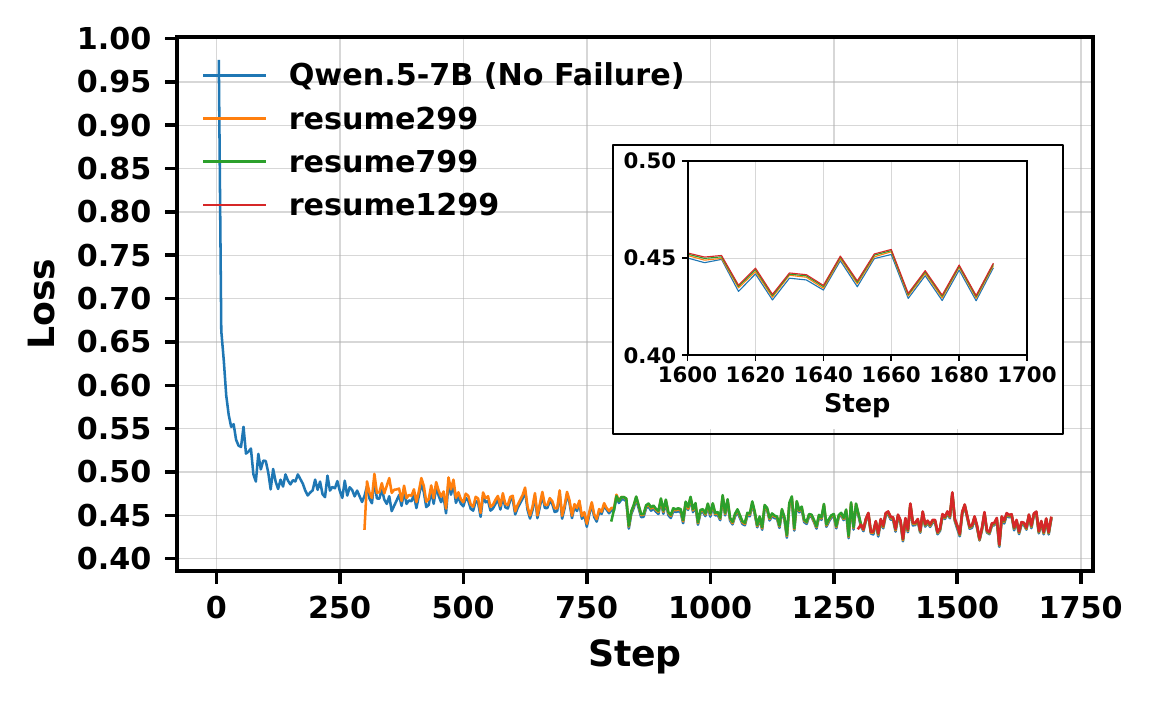}
        \caption{Loss trajectories under cascading composite recoveries (fixed $K=10^{-3}$).}
        \label{fig:loss_overlay}
    \end{subfigure}
    \hspace{1em}
    \begin{subfigure}[b]{0.45\textwidth}
        \centering
        \includegraphics[width=\linewidth]{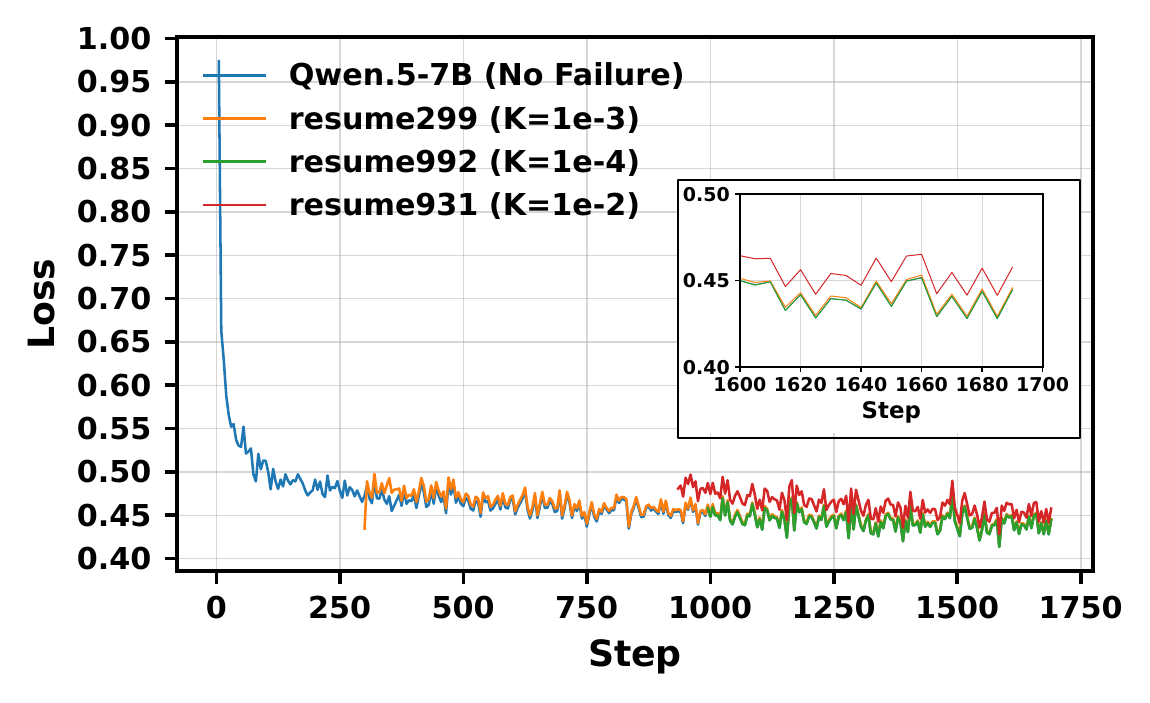}
    \caption{Loss trajectories under varying drift thresholds K (single failure).}
    \label{fig:loss_ablation}
    \end{subfigure}
    \caption{Recovery fidelity of LayerCheck on Qwen2.5-7B (OpenThoughts dataset).}
    \label{fig:loss}
\end{figure*}

The most distinctive aspect of \texttt{LayerCheck} is composite recovery: after a failure, the recovered training state is assembled by choosing, for each layer, the most recent persisted version of that layer's weights and associated optimizer tensors.
Because layers are persisted at different iterations, the reconstructed checkpoint inevitably contains mixed timestamps.
This raises a concrete risk: even if each individual layer is internally consistent (weights and optimizer moments originate from the same time), the overall model state may represent a configuration that never existed during failure-free training.

\texttt{LayerCheck} addresses this risk with an explicit system invariant.
Rather than attempting to force all layers to share the same timestamp at recovery, it enforces a bounded-staleness guard $S_{\max}$ that ensures no layer can be arbitrarily stale relative to the failure point.
This invariant is inspectable during runtime and is also the key assumption used in our convergence discussion (\S\ref{sec:convergence_layercheck}).

To evaluate recovery fidelity, we examine the training loss trajectories after restart. A mixed-timestamp checkpoint can perturb both parameters and optimizer moments; if this perturbation were large or adversarial, we would expect an instability divergence, or a widening gap indicating that recovery pushes training into a different regime.

For Qwen2.5-7B fine-tuning on OpenThoughts, we inject failures at the point of \emph{maximum parameter-weighted average staleness} for each configuration.

This is the most globally challenging moment for composite recovery because it maximizes the checkpoint-wide mismatch across layers while the maximum per-layer staleness still remains bounded by $S_{\max}$.
For $K{=}10^{-3}$, $10^{-2}$, and $10^{-4}$, these points occur at steps 299, 931, and 992, respectively.
We use the failure-free run as the visual reference, since exact full-state (lossless) recovery reproduces the failure-free trajectory by construction.

Figure~\ref{fig:loss_overlay} first examines a fixed operating point ($K{=}10^{-3}$) under cascading composite recoveries. The first failure is injected at step~299, followed by two additional failures at steps~799 and~1299, with training resuming from the reconstructed composite checkpoint after each one. The recovered trajectory remains close to the failure-free baseline throughout, with no visible sign of compounding instability or progressive drift. We then measure the post-restart training's final loss: the maximum absolute deviation from the failure-free trajectory is $0.0024$ (\textbf{0.54\%} relative), well within step-to-step training noise. This result suggests that the perturbation introduced by mixed-timestamp recovery does not accumulate into runaway optimization error across repeated failures, consistent with the multi-failure analysis in~\S\ref{sec:convergence_layercheck}.

Figure~\ref{fig:loss_ablation} next shows how recovery fidelity changes with the drift threshold $K$ under single-failure recovery.
Across all three thresholds, the resumed runs closely track the baseline after restart.
The patterns also match the intended trade-off.
Smaller $K$ refreshes layers more aggressively, so the post-recovery deviation is smaller.
Larger $K$ skips more writes, so the recovered state can be more stale and may exhibit a slightly higher loss band immediately after restart; however, this deviation remains bounded rather than amplifying with continued optimization.

Taken together, Figure~\ref{fig:loss} shows that LayerCheck recovery is robust in two complementary senses: it remains stable under repeated failures at a fixed operating point, and it remains well behaved across the freshness--efficiency trade-off induced by different $K$ values.
This is exactly the behavior the $S_{\max}$ guard is designed to enable.

\subsection{Checkpoint Efficiency: End-to-End Time and Storage}
\label{subsec:eval-io}

We next quantify steady-state checkpoint cost during model training.
Because \texttt{LayerCheck} targets redundant persistence of slow-changing layers, we evaluate an intentionally I/O-stressful checkpoint-per-iteration regime, as also used by prior checkpointing systems~\cite{wangGEMINIFastFailure2023,Datastate-LLM,lowdiff}. We use this setting to evaluate whether each system can sustain high-frequency checkpointing and to expose differences in written bytes, serialization work, and overlap behavior, rather than modeling every production workload.

\subsubsection{Time overhead}
Figure~\ref{fig:e2e-time} reports end-to-end runtime for 1 epoch on OpenThoughts. \texttt{LayerCheck} achieves the shortest training time across all three models, with overhead of only 11.7--20.0\% relative to checkpoint-free training ($t_{\text{net}}$), \textbf{1.31$\times$} faster than LowDiff on Llama3.2-1B and 3.76$\times$ faster than GEMINI on Qwen2.5-7B, with larger gaps for CheckFreq and DeepSpeed default.

Checkpoint fraction alone does not explain this. \texttt{LayerCheck}'s checkpoint fraction (10.7--12.0\%) is comparable to LowDiff's (9.3--15.9\%), yet end-to-end times differ substantially. The reason is that under per-iteration checkpointing, the baselines must drain full-model-scale checkpoint work each step. Once this work exceeds one training iteration, overlap breaks down and visible stalls emerge.
This is exactly where \texttt{LayerCheck} gains its systems advantage: under checkpoint pressure, it is not enough to pipeline checkpointing better; reducing the amount of state and serialization work itself is crucial.

Crucially, this overhead is dominated by checkpoint I/O (the hatched portion of each LayerCheck bar). Subtracting the checkpoint time, the remaining training time essentially matches $t_{\text{net}}$ across all three models, indicating that LayerCheck's per-iteration drift tracking adds negligible overhead to the optimizer update path itself. This is consistent with our piggyback design (\S\ref{subsec:implementation}), where drift computation reuses tensors already materialized during the AdamW update rather than introducing additional passes over model parameters.

The shared Lustre backend further amplifies these effects, but the qualitative trend is not specific to Lustre.
Methods that write less state and avoid repeated full-model rewrites are less exposed to storage variability and I/O backpressure.

Beyond reducing steady-state training stalls, this result also has an important implication for failures.
By lowering the cost of frequent persistence, \texttt{LayerCheck} makes it more practical to keep restart points close to the failure point.
We therefore examine recovery overhead under comparable recovery quality in the next subsection.

\begin{figure}[ht!]
    \centering
      \includegraphics[width=0.8\linewidth]{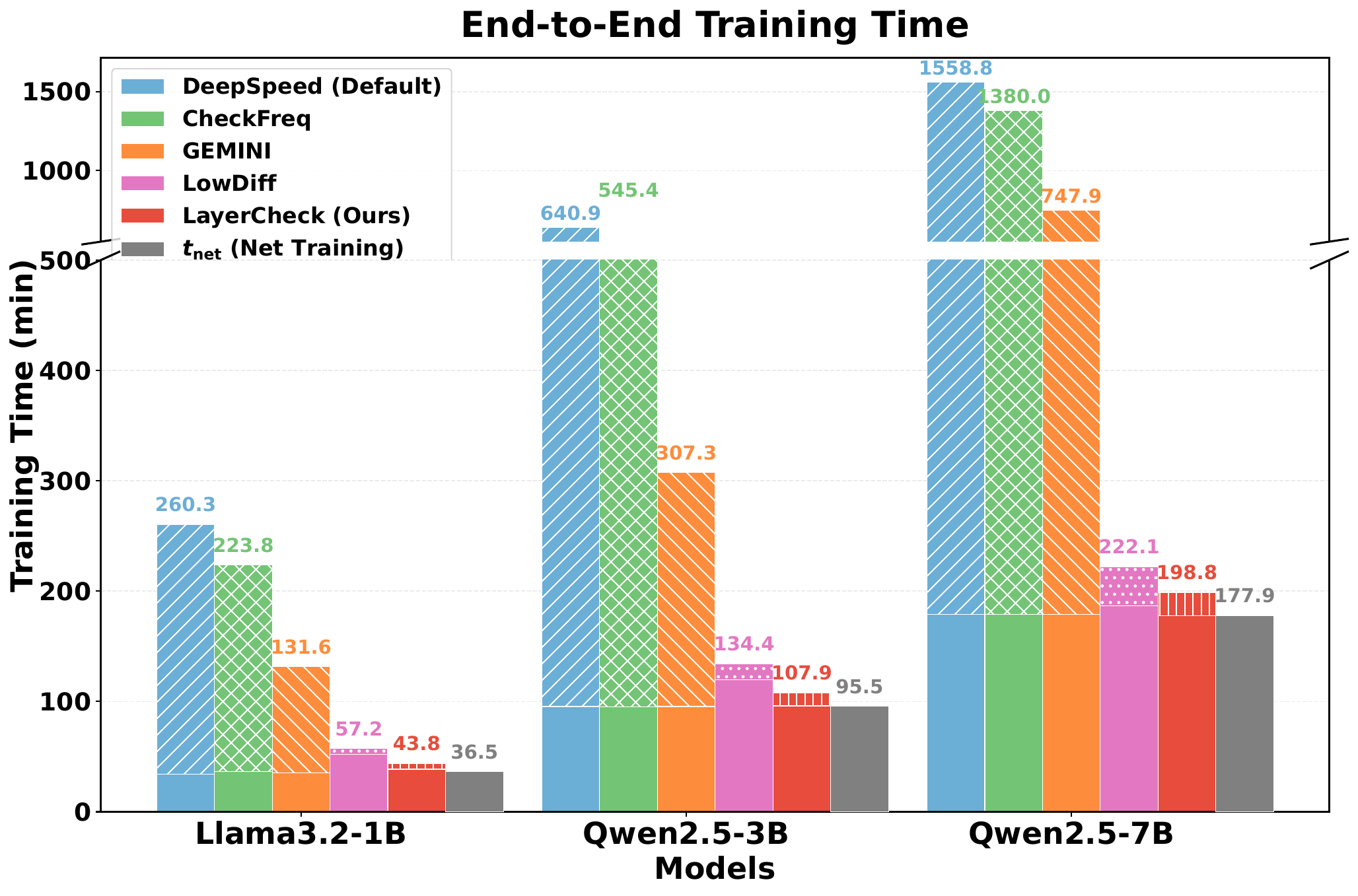}
      \caption{End-to-end training time for 1 epoch on OpenThoughts dataset with checkpoint-per-iteration. Hatched portions denote non-overlapped checkpointing time in training.}
      \label{fig:e2e-time}
\end{figure}

\subsubsection{I/O volume}
We further quantify I/O volume using checkpoint storage footprint.
To make this concrete, we report total checkpoint size over the run, the number of checkpoint events (directories), and the mean checkpoint size.
These capture complementary aspects of storage pressure: total footprint reflects cumulative write volume, while mean checkpoint size relates to per-event burstiness that affects latency and overlap.

Table~\ref{tab:ckpt_methods_openthoughts} reports checkpoint statistics on OpenThoughts for $K=10^{-3}$.
The ``\#Ckpt Layers'' column counts the total number of \emph{layer writes} across the run, highlighting the key distinction from full-state checkpointing where each checkpoint necessarily rewrites every layer.
As expected, \texttt{LayerCheck} produces many checkpoint events that contain only a small subset of layers.

The reductions in written bytes are substantial.
\texttt{LayerCheck} reduces total checkpoint storage by up to \textbf{22.6×} (average $\sim 17.1\times$) relative to LowDiff, and reduces mean per-checkpoint size by up to 6.6× (average $\sim 5.0\times$).
This result directly supports the paper's main systems argument: layer-level sparsity in ``meaningful change'' can be exploited to avoid rewriting most of the model most of the time.
In turn, fewer bytes translates to fewer stalls in two ways.
It reduces raw write time, and it increases the likelihood that persistence work can be pipelined without blocking the training loop, especially on storage backends with high variance in throughput.

\begin{table}[t]
\centering
\footnotesize
\setlength{\tabcolsep}{3pt}
\renewcommand{\arraystretch}{1.05}
\caption{Checkpoint statistics on OpenThoughts ($K=10^{-3}$). For each model, our method is listed last and highlighted in bold for readability.}
\label{tab:ckpt_methods_openthoughts}
\begin{tabular}{llrrrr}
\toprule
\makecell[l]{Model (Layers)} &
Method &
\makecell[r]{Ckpt Layers} &
\makecell[r]{Ckpt Dirs} &
\makecell[r]{Total Size (GB)} &
\makecell[r]{Mean Size (GB)} \\
\midrule
\makecell[l]{Qwen2.5-7B\\(31)}  & Default                 & 52452            & 1692            & 180367.20         & 106.60          \\
                                  & LowDiff                 & --               & 930             & 19258.68          & 20.71           \\
                                  & \textbf{LayerCheck}    & \textbf{319}    & \textbf{277}   & \textbf{1205.38} & \textbf{4.35}  \\
\midrule
\makecell[l]{Qwen2.5-3B\\(38)}  & Default                 & 64296            & 1692            & 73111.32          & 43.21           \\
                                  & LowDiff                 & --               & 930             & 8512.85           & 9.15            \\
                                  & \textbf{LayerCheck}    & \textbf{340}    & \textbf{273}   & \textbf{376.21}  & \textbf{1.38}  \\
\midrule
\makecell[l]{Llama3.2-1B\\(18)} & Default                 & 30456            & 1692            & 29305.44          & 17.32           \\
                                  & LowDiff                 & --               & 930             & 3837.05           & 4.12            \\
                                  & \textbf{LayerCheck}    & \textbf{266}    & \textbf{242}   & \textbf{264.82}  & \textbf{1.02}  \\
\bottomrule
\end{tabular}
\end{table}

Finally, we examine how the threshold $K$ controls the efficiency--freshness trade-off.
Table~\ref{tab:layercheck_ablation_datasets_qwen7b} reports the footprint across datasets for Qwen2.5-7B.
A larger $K$ persists fewer layers and yields fewer checkpoints, while a smaller $K$ refreshes layers more aggressively and increases total footprint.
We use $K=10^{-3}$ as the default operating point because it provides a practical balance between recovery fidelity and checkpoint cost in our evaluated workloads.

\begin{table}[t]
\centering
\footnotesize
\caption{\texttt{LayerCheck} ablation across datasets for Qwen2.5-7B (31 layers).}
\label{tab:layercheck_ablation_datasets_qwen7b}
\begin{tabular}{llrrrr}
\toprule
Dataset & $K$ &
\makecell[r]{Ckpt Layers} &
\makecell[r]{Ckpts} &
\makecell[r]{Total Size (GB)} &
\makecell[r]{Mean Size (GB)} \\
\midrule
OpenThoughts & $10^{-2}$ & 94   & 40   & 393.20   & 9.83 \\
             & $10^{-4}$ & 2398 & 1208 & 9375.56  & 7.76 \\
             & \textbf{$10^{-3}$} & 319  & 277  & 1205.38  & 4.35 \\
\midrule
MedQA        & $10^{-2}$ & 96   & 52   & 355.43   & 6.84 \\
             & $10^{-4}$ & 2821 & 1378 & 11388.00 & 8.26 \\
             & \textbf{$10^{-3}$} & 373  & 327  & 1452.00  & 4.44 \\
\bottomrule
\end{tabular}
\end{table}

\subsection{Recovery Overhead}
\label{subsec:eval-recovery}

Recovery overhead is meaningful only when recovery quality is controlled.
A method can appear ``fast'' simply by checkpointing much more frequently, while infrequent checkpointing can appear ``slow'' because it must replay more lost iterations after a failure. Per-iteration checkpointing imposes prohibitive overhead on the baselines, and even setting that aside, all methods would replay at most one step after a failure, making the recovery comparison uninformative. We therefore use a matched-freshness regime instead. LowDiff and \texttt{LayerCheck} retain their native per-step persistence semantics, since both mechanisms are defined around continual step-level state preservation. For rollback-based baselines (DeepSpeed default, CheckFreq, and GEMINI), we set the checkpoint interval to match the average staleness observed under \texttt{LayerCheck}, which is 106 training steps on Qwen2.5-7B. Assuming failures are uniformly distributed within a checkpoint interval, this corresponds to an average rollback distance of 53 steps after a failure.

We decompose the total recovery time for a single failure as
\begin{equation}
\label{eq:recovery-time-def}
    t_{\mathrm{recovery}} = t_{\mathrm{load}} + t_{\mathrm{retrain}} + t_{\mathrm{reconstruct}},
\end{equation}
where $t_{\mathrm{load}}$ is the time to read checkpoint state from storage, $t_{\mathrm{retrain}}$ is the time to replay lost iterations after rollback, and $t_{\mathrm{reconstruct}}$ is method-specific reconstruction cost (e.g., layer merging or decompression).

In Figure~\ref{fig:decomp}, we report the aligned recovery overhead
\begin{equation}
\label{eq:recovery-overhead-aligned}
    t_{\mathrm{recovery}}^{\star} = t_{\mathrm{load}} + t_{\mathrm{retrain}},
\end{equation}

We exclude \(t_{\mathrm{reconstruct}}\) from the cross-method comparison because GEMINI and LowDiff do not provide complete public recovery implementations, and reconstruction cost is highly implementation-dependent. This narrows the scope of the comparison but preserves comparability across frameworks.

Figure~\ref{fig:decomp} reports the aligned recovery overhead across all three models under this matched-freshness setting. \texttt{LayerCheck} achieves the lowest recovery overhead at 10s, compared to 34s for LowDiff (\textbf{3.4×} faster) for Qwen2.5-7B.
Under this protocol, the rollback-based full-state methods naturally become similar, because they replay the same expected rollback window and differ mainly in where these state loaded.
\texttt{LayerCheck} achieves the lowest recovery overhead among the compared methods.
This improvement follows the same logic as the steady-state results: by writing less state during training, \texttt{LayerCheck} has less state to load after failures.
In addition, composite recovery reduces the need to roll back and replay iterations; it is designed to restart near the failure point under the bounded-staleness invariant.
LowDiff also benefits from per-step recovery semantics, but \texttt{LayerCheck} remains lower because selective persistence reduces the amount of state that must be materialized at restart than the differential checkpoint.

\begin{figure}[ht!]
    \centering
    \includegraphics[width=0.8\linewidth]{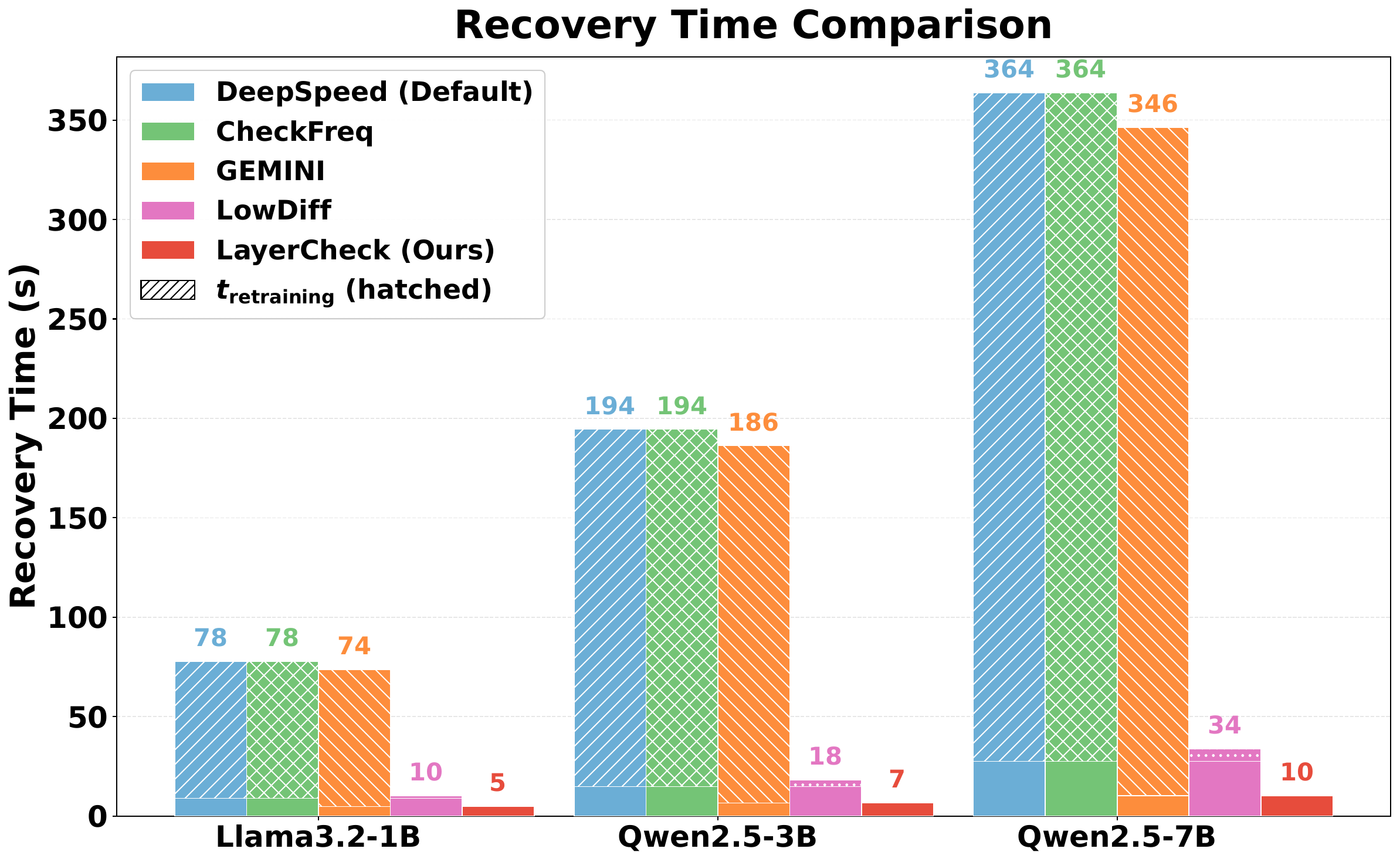}
    \caption{Aligned recovery overhead of a single failure under matched restart freshness across three models. We report $t_{\mathrm{load}} + t_{\mathrm{retrain}}$ as $t_{\mathrm{recovery}}^{\star}$.}
    \label{fig:decomp}
\end{figure}

\subsection{Scalability under Repeated Failures}
\label{subsec:eval-scalability}

As training scales out, failures become more frequent~\cite{wangGEMINIFastFailure2023}, and steady-state checkpoint overhead alone no longer predicts overall system efficiency.
We therefore study projected behavior under repeated failures using a data-driven simulation parameterized by measured system components in the previous sections, with mean time between failures (MTBF) from 0.01 to 3 hours. All method-specific time components used in the simulation (e.g., per-iteration training time, checkpoint overhead, and recovery time terms in \eqref{eq:recovery-time-def}) are measured on Qwen2.5-7B using the same OpenThoughts dataset under the same storage stack, and then plugged into the MTBF-driven failure injection model.

Following GEMINI~\cite{wangGEMINIFastFailure2023}, we report effective training time ratio (ETTR), defined as
\begin{equation}
\label{eq:ettr}
\mathrm{ETTR} \triangleq \frac{t_{\mathrm{train}}}{t_{\mathrm{train}} + t_{\mathrm{ckpt}} + t_{\mathrm{recovery}}}.
\end{equation}
ETTR can be read as utilization: the fraction of wall-clock time spent doing useful forward/backward/optimizer work rather than checkpointing and recovery.

The simulation in Figure~\ref{fig:scaling} shows that, \texttt{LayerCheck} consistently achieves the highest ETTR across all MTBF values.
For example, at MTBF=0.5 hours, \texttt{LayerCheck} attains \textbf{92.3\%}, outperforming LowDiff (88.8\%) and GEMINI (84.5\%).
This result ties together the earlier sections in a single metric: when failures are frequent enough to matter, the system that minimizes both steady-state I/O stalls and failure-time recovery costs yields the highest effective throughput.

\begin{figure}[ht!]
    \centering
    \includegraphics[width=0.8\linewidth]{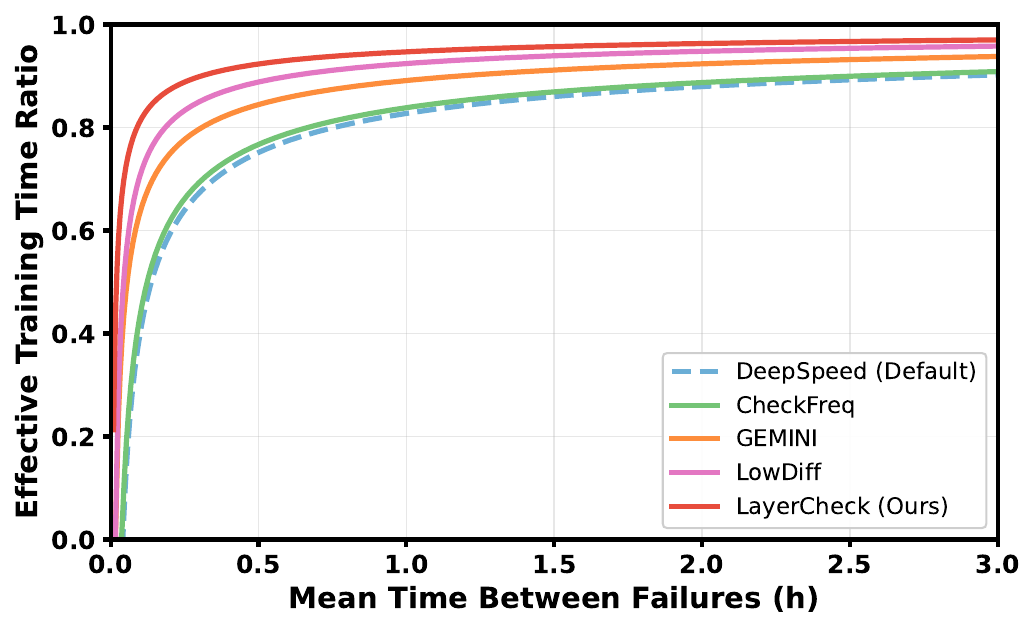}
    \caption{Effective training time ratio under repeated failures with varying MTBF. Results are computed from Eq.~\ref{eq:ettr} with timing parameters measured on Qwen2.5-7B (OpenThoughts).}
    \label{fig:scaling}
\end{figure}

\subsection{Recovered Model Quality}
\label{subsec:eval-quality}

Loss trajectories establish that training remains stable after composite recovery, but they do not fully answer whether the final model behaves the same on downstream evaluation tasks.
This matters because LLM fine-tuning losses can correlate imperfectly with benchmark scores.

Table~\ref{tab:benchmark_all_merged} therefore compares final benchmark scores for (i) failure-free fine-tuning and (ii) fine-tuning with an injected failure and recovery followed by continued training.
Across datasets and model sizes, recovered runs match the failure-free results closely and do not exhibit systematic degradation.
Combined with Figure~\ref{fig:loss}, these results support the paper's main claim: \texttt{LayerCheck} reduces checkpoint and recovery cost while preserving the end-to-end outcome of training.

\begin{table*}[t]
\centering
\caption{Benchmark on OpenThoughts and MedQA (After fine-tuning vs. recovered checkpoints).}
\label{tab:benchmark_all_merged}
\footnotesize
\resizebox{\linewidth}{!}{
\begin{tabular}{lllrrrrrrr}
\toprule
Dataset & Model & Status & ARC-easy & HellaSwag & Lambada & PIQA & MedMCQA & MMLU-Med & PubMedQA \\
\midrule
\multirow{6}{*}{MedQA}
 & \multirow{2}{*}{Qwen2.5-7B}  & After FT (No failure)           & 81.44 & 60.59 & 67.07 & 79.43 & 60.39 & 92.00 & 75.20 \\
 &                             & \textbf{Recovered (step=888)}      & \textbf{81.78} & 60.36 & \textbf{67.30} & \textbf{80.03} & \textbf{60.41} & \textbf{93.00} & \textbf{75.60} \\
\cmidrule(lr){2-10}
 & \multirow{2}{*}{Qwen2.5-3B}  & After FT (No failure)           & 77.65 & 54.82 & 63.50 & 77.20 & 52.12 & 82.00 & 74.60 \\
 &                             & \textbf{Recovered (step=867)}      & \textbf{77.95} & \textbf{54.90} & \textbf{63.56} & \textbf{77.31} & \textbf{52.26} & 81.00 & 74.60 \\
\cmidrule(lr){2-10}
 & \multirow{2}{*}{Llama3.2-1B} & After FT (No failure)           & 64.02 & 46.73 & 53.64 & 73.12 & 31.77 & 33.00 & 55.80 \\
 &                             & \textbf{Recovered (step=1374)}     & \textbf{64.14} & \textbf{46.91} & 53.60 & 72.96 & \textbf{32.44} & \textbf{34.00} & 55.80 \\
\midrule
\multirow{6}{*}{OpenThoughts}
 & \multirow{2}{*}{Qwen2.5-7B}  & After FT (No failure)           & 79.92 & 59.70 & 68.87 & 77.97 & 59.86 & 85.00 & 75.20 \\
 &                             & \textbf{Recovered (step=299)}      & \textbf{80.35} & 59.63 & 68.37 & \textbf{78.24} & 59.74 & 85.00 & 75.00 \\
\cmidrule(lr){2-10}
 & \multirow{2}{*}{Qwen2.5-3B}  & After FT (No failure)           & 77.31 & 54.46 & 66.45 & 78.07 & 51.78 & 74.00 & 73.00 \\
 &                             & \textbf{Recovered (step=775)}      & 77.15 & 54.38 & \textbf{66.56} & \textbf{78.18} & 51.66 & 74.00 & 72.60 \\
\cmidrule(lr){2-10}
 & \multirow{2}{*}{Llama3.2-1B} & After FT (No failure)           & 64.77 & 47.39 & 60.22 & 74.54 & 31.44 & 28.00 & 58.00 \\
 &                             & \textbf{Recovered (step=1254)}     & 64.69 & 47.39 & \textbf{60.43} & 74.54 & \textbf{31.46} & 28.00 & 58.00 \\
\bottomrule
\end{tabular}
}
\end{table*}

\subsection{Summary and Limitations}
\label{subsec:eval-summary}

Taken together, these experiments support a consistent picture. Composite recovery does not destabilize optimization, which is the key stability concern raised by mixed timestamps.
Selective persistence substantially reduces written bytes and, in an I/O-stressful regime, translates to end-to-end time savings. Under an aligned staleness budget, \texttt{LayerCheck} also reduces recovery overhead, and these benefits compound under repeated failures as reflected in ETTR.

Two practical limitations remain.
First, the absolute magnitude of improvements depends on the storage backend~\cite{Qosflow}.
Faster local storage reduces checkpoint cost for all methods, but selective persistence still reduces bytes written and thus reduces the sensitivity of training throughput to storage variability~\cite{dataflow,carat}.

Second, we exclude method-specific reconstruction time $t_{\mathrm{reconstruct}}$ from the cross-method comparison. This term is primarily implementation-dependent and is not reported uniformly across prior systems; moreover, several baselines do not provide complete public recovery code. We therefore compare only the recovery components that can be aligned consistently across methods. \texttt{LayerCheck} nevertheless includes an explicit reconstruction path together with the $S_{\max}$ fidelity guard, making its recovery semantics explicit.

Finally, our implementation, evaluation, and analysis focus on AdamW. LayerCheck’s persistence decision is based on realized layer-wise parameter changes rather than Adam-specific statistics, so the same principle can be applied conceptually to vanilla SGD. In this case, LayerCheck would persist the selected layer weights without per-parameter optimizer state. Because the optimizer affects the scale and temporal pattern of parameter updates, the base threshold \(K_0\) should be calibrated accordingly. A smaller \(K_0\) triggers more frequent persistence, reducing layer staleness and improving recovery fidelity at the cost of higher I/O overhead, whereas a larger \(K_0\) reduces checkpoint writes but may increase staleness and recovery deviation.

%% file: related_work.tex
\section{Related Work}\label{sec:related}

\textbf{Checkpointing for Deep Learning and LLM Training.}
Checkpointing is the standard fault-tolerance mechanism in major deep learning frameworks such as PyTorch~\cite{paszkePyTorchImperativeStyle2019} and TensorFlow~\cite{abadiTensorFlowSystemLargescale2016}, but it exposes the familiar trade-off between checkpoint overhead and recovery cost. Prior work can be grouped into several lines. \emph{Pipeline and concurrent checkpointing} reduce training stalls by overlapping snapshotting, and persistence with computation, as in CheckFreq~\cite{checkfreq}, GEMINI~\cite{wangGEMINIFastFailure2023}, PCcheck~\cite{PCcheck}, Datastate-LLM~\cite{Datastate-LLM}, and FlowCheck~\cite{flowchcek}. \emph{Elastic or redundancy-aware recovery} improves resilience under failures or changing resources, as in CPR~\cite{maengCPRUnderstandingImproving}, Just-in-Time checkpointing~\cite{guptaJustInTimeCheckpointingLow2024}, Elastor~\cite{Elastor}, and Checkmate~\cite{Checkmate}. 
\emph{Partial, incremental, and differential checkpointing} reduce what must be written or reconstructed, including SCAR~\cite{partialrecover}, Check-N-Run~\cite{check-n-run}, LowDiff~\cite{lowdiff}, LLMTailor~\cite{tailor} and Amber~\cite{Amber}. Among these systems, LowDiff~\cite{lowdiff} is a recent online differential-checkpointing baseline and is therefore our primary comparison on the online-I/O axis. LLMTailor is closest in spirit: it assembles a resumable checkpoint from existing ones via fixed, rule-based layer-merging recipes, with no online layer selection or convergence analysis—its filtered variants can even degrade model quality.

In contrast, LayerCheck makes online, per-layer, per-step persistence decisions via accumulated drift with LR-aware thresholding, provides a bounded-staleness convergence guarantee for mixed-timestamp recovery, and avoids replay-based recovery while still reducing checkpoint I/O.

\textbf{Checkpoint Compression and Model Reduction.}
Some prior work reduces checkpoint size through compression. ExCP~\cite{liexcp} combines quantization and compression to shrink checkpoints, while Inshrinkerator~\cite{inshrinkerator} exploits weight sensitivity to compress checkpointed states. QD-Compressor~\cite{QD} further applies layer-wise adaptive quantization to model snapshots in federated learning. AutoCheck~\cite{fu2024autocheck} and ADVICE~\cite{huang2026advice} propose an LLVM-based compiler augmentation to identify variables necessary to checkpointing that significantly reduces checkpoint states. ZOCheck~\cite{sun2026zocheck} designs a
lightweight checkpointing inspired by LLM Zeroth-order (ZO) execution semantics, representing training progress with compressed states and recovers via deterministic replay. 
These methods reduce storage footprint, but they are primarily designed for local checkpoints or smaller federated settings rather than online checkpointing in large-scale distributed LLM training.

More broadly, model reduction techniques such as quantization, pruning, and structural compression have been used to lower the storage and compute cost of large models. Recent quantization studies~\cite{QQQ,qserve} show that low-bit representations can preserve a favorable accuracy--efficiency trade-off at inference time. Structural approaches including LLM-Streamline~\cite{streamlining}, LLM-Pruner~\cite{Llm-pruner}, LaCo~\cite{yangLaCoLargeLanguage2024}, and ShortGPT~\cite{menShortGPTLayersLarge2024} reduce model size by pruning, or merging redundant layers, building on the broader pruning literature~\cite{hanLearningBothWeights2015}. Although these methods are not checkpointing systems, they suggest that LLM layers differ in redundancy and sensitivity. \texttt{LayerCheck} is motivated by a related observation, but applies it to fault tolerance: instead of compressing the deployed model itself, it selectively persists full layer states online during training to reduce checkpoint I/O and recovery cost.

%% file: conclusion.tex
\section{Conclusion}\label{sec:conclude}
In this paper, we introduce \texttt{LayerCheck}, an adaptive layer-wise system for reducing checkpoint overhead in LLM post-training. By selectively persisting only layers with significant accumulated drift, \texttt{LayerCheck} reduces checkpoint I/O and shortens end-to-end training time while preserving final model quality. By reusing already-available training signals, \texttt{LayerCheck}'s layer-drift tracking adds negligible per-step overhead, while selective persistence substantially reduces the checkpoint cost. Experiments show that our approach reduces the total checkpoint size by up to 22.6× and the checkpointing time by up to 1.31× compared with existing SOTA, while preserving final model quality. A key next step is to integrate \texttt{LayerCheck} with overlap-oriented checkpoint pipelines, so that our reduced checkpoint work can also benefit from end-to-end overlap across computation, communication, and I/O. 